\documentclass{fundam}

\usepackage{url} 
\usepackage[ruled,lined]{algorithm2e}
\usepackage{graphicx}
\graphicspath{ {Figures/} }
\usepackage{tikz,esvect}
\usetikzlibrary{automata, positioning, arrows}

\usepackage{hyperref}

\begin{document}

\setcounter{page}{301}
\publyear{2021}
\papernumber{2075}
\volume{182}
\issue{3}

 \finalVersionForARXIV

\title{Query-Points Visibility Constraint Minimum Link Paths\\  in Simple Polygons}

\author{Mohammad Reza Zarrabi\thanks{Address for correspondence: Faculty of Electrical Engineering and Computer Science,
                          Tarbiat Modares University, Tehran, Iran. \newline \newline
          \vspace*{-6mm}{\scriptsize{Received  September 2020; \ revised August 2021.}}}
         \\
Department of Electrical and Computer Engineering \\
Tarbiat Modares University \\
Tehran, Iran\\
m.zarabi@modares.ac.ir
\and
  Nasrollah Moghaddam Charkari\\
  Department of Electrical and Computer Engineering \\
Tarbiat Modares University \\
Tehran, Iran
 } 

\maketitle

\runninghead{M.R. Zarrabi and N.M. Charkari}{Query-Points Visibility Constraint Minimum Link Paths in Simple Polygons}

\begin{abstract}
We study the query version of
constrained minimum link paths between two points inside a simple polygon $P$ with $n$ vertices
such that there is at least one point on the path, visible from a query point.
The method is based on partitioning $P$ into a number of faces of equal link distance from a point, called a link-based shortest path map (SPM).
Initially, we solve this problem for two given points $s$, $t$ and a query point $q$.
Then, the proposed solution is extended to a general case for
three arbitrary query points $s$, $t$ and $q$.
In the former, we propose an algorithm with $O(n)$ preprocessing time.
Extending this approach for the latter case, we develop an algorithm with $O(n^3)$ preprocessing time.
The link distance of a $q$-$visible$ path between $s$, $t$ as well as the path
are provided in time $O(\log n)$ and $O(m+\log n)$, respectively, for the above two cases, where $m$ is the number
of links.
\end{abstract}


\section{Introduction}
One of the problems in the field of Robotics and Computational Geometry is finding a
\emph{minimum link path} between two points in a simple polygon.
A minimum link path between two points $s$ and $t$ is a chain of line segments (\emph{links})
connecting them inside a simple polygon $P$ with $n$ vertices that has the minimum number of links.
The \emph{link distance} between $s$ and $t$ is defined as the number of links in a minimum link path.
Finding a minimum link path between two \emph{fixed} points inside
a simple polygon was first studied by Suri \cite{Suri_1986}.
He introduced an $O(n)$ time algorithm for this problem.
Afterwards, Ghosh \cite {Ghosh_1991} presented an alternative algorithm, which also
runs in $O(n)$ time.
To solve this problem in a polygonal domain, Mitchell
et al. \cite {Mitchell_1992} proposed an incremental algorithm that runs in time $O(n^2 \log^2 n)$, where
$n$ is the total number of edges of the obstacles.
On the other hand, a more general framework was established for minimum link paths by Suri \cite {Suri_1990}
based on \emph{Shortest Path Map} \emph{(SPM)}.
By the construction of SPM from a fixed point, the simple polygon $P$ is divided into \emph{faces} of equal
link distance from that point in linear time.
With this property, Arkin et al. \cite {Arkin_1995} developed an algorithm for computing the link distance
between two arbitrary \emph{query} points inside $P$.
The algorithm computes window partitioning (SPM) of $P$
from every vertex of $P$ and from every extension point
of the visibility graph of $P$.
It can be seen that endpoints of these windows divide edges of $P$ into $O(n^2)$ \emph{atomic} segments.
These segments satisfy the property that the \emph{combinatorial type} of $SPM(x)$ is the same for all points $x$ in the interior of each segment.
Thus, the algorithm requires $O(n^3)$ preprocessing time
and answers a link distance query in $O(\log n)$ time.

In many applications, it is required for a robot to have direct visibility from a viewpoint during its motion
\cite {Chen_1996}.
Some examples are moving guards, resource collectors, wireless communications, etc.
Minimum link paths have important applications in Robotics
since turns are costly while straight line movements are inexpensive in robot motion planning.
In the minimum link paths problem with point visibility constraint, the aim is to find a minimum link path between two points $s$ and $t$ such that there is at least one point on the path from which a given viewpoint $q$ is visible
(a $q$-$visible$ path \cite {Zarrabi_2019}).

The constrained version of minimum link paths problem
was studied in \cite {Zarrabi_2019} for three fixed points $s$, $t$ and $q$ inside a simple polygon.
In this paper, we study the query version of the problem inside a simple polygon $P$ with $n$ vertices for two cases.
First, suppose that two fixed points $s$ and $t$ are given, the goal is to find a $q$-$visible$ path
between $s$ and $t$ for an arbitrary query point $q$.
Second, we consider the same problem for three arbitrary query points $s$, $t$ and $q$.
We propose two algorithms with $O(n)$ and $O(n^3)$ preprocessing time for the above cases, respectively.
The link distance of a $q$-$visible$ path between $s$, $t$ as well as the path
are provided in time $O(\log n)$ and $O(m+\log n)$, respectively, for both cases, where $m$ is the number
of links.

Similarly, a $q$-$visible$ path is defined for the Euclidean metric.
In this case, the query version of the shortest Euclidean path problem for two fixed points $s$, $t$ and a query point $q$
inside a simple polygon with $n$ vertices was studied in \cite {Khosravi_2007}.
The given algorithm preprocesses the input in $O(n^3)$ time and provides $O(\log n)$ query time.
A simpler form of the problem was studied in \cite {Khosravi_2005}
in which the goal is to find the shortest Euclidean path from $s$ to view $q$
(without going to a destination point).
The algorithm requires $O(n^2)$ preprocessing time
and answers a query in $O(\log n)$ time.
Arkin et al. \cite {Arkin_2016} improved the preprocessing time to $O(n)$ for simple polygons.
Also, they built a data structure of size $O(n^2 2^{\alpha(n)} \log n)$ that can answer each query in $O(n \log^2 n)$ time
for a polygonal domain with $h$ holes and $n$ vertices, where $\alpha(n)$ is the inverse Ackermann function.
Recently, a new data structure of size $O(n \log h + h^2)$ was presented for this problem
that can answer each query in $O(h \log h \log n)$ time \cite {Wang_2019}.

The main differences between approaches in the Euclidean metric and the link distance metric are as follows.
Optimal paths that are unique under the Euclidean metric need not be unique under the link distance metric.
Also, Euclidean shortest paths only turn at reflex vertices while minimum link paths can turn anywhere.
Thus, minimum link paths problems are usually more difficult to solve than equivalent Euclidean shortest path problems.

The main idea of the proposed algorithms given in this paper is to consider an \emph{edge} of the visibility polygon
(called $e_q$)
for a query point $q$ as a separator chord inside $P$,
i.e., if $s$ and $t$ lie in different sides of such a chord, an optimal link path between them will be the answer. Otherwise, a $q$-$visible$ path should have a non-empty intersection with the side of the chord containing $q$.
Therefore, the problem can be reduced to find an appropriate edge and optimal contact
points between a $q$-$visible$ path and the other side of the edge.
To answer the queries efficiently, we preprocess the input
using \emph{map overlay} \cite {Finke_1995}, \emph{point location} \cite {Edelsbrunner_1986},
\emph{ray shooting} \cite {Guibas_1987} and \emph{shortest path map} \cite {Suri_1990} techniques.

In Section 2, we introduce the problem definition and notation.
Section 3 gives the basic lemmas and definitions.
Section 4 shows the main idea and flow of the algorithm.
Section 5 describes our algorithm for \emph{single} query point, and Section 6
generalizes this algorithm to \emph{triple} query points.
Section 7 concludes with some open problems.

\section{Problem definition and notation}

Let $P$ be a simple polygon in the plane with $n$ vertices.
For three points $s$, $t$ and $q$ inside $P$,
the goal is to preprocess the input to answer
two types of queries:
\begin{enumerate}
\itemsep=0.95pt
\item [{$1)$}] Given a query point $q$, find all $q$-$visible$ paths between
fixed points $s$ and $t$ in $P$ (single query).
\item [{$2)$}] Given three query points $s$, $t$ and $q$,
find a $q$-$visible$ path between $s$ and $t$ in $P$ (triple query).
\end{enumerate}

We use the following notation throughout the paper:
\begin{enumerate}
\itemsep=0.95pt
\item [{$\bullet$}]  $V(x):$ the visibility polygon of a point $x$ $\in$ $P$
\item [{$\bullet$}]  $\pi_L(x,y):$ a minimum link path from a point $x$ to a point $y$ inside $P$
\item [{$\bullet$}]  $\pi_E(x,y):$ the shortest Euclidean path from a point $x$ to a point $y$ inside $P$
\item [{$\bullet$}]  $MLP(x,y,q):$ a $q$-$visible$ (minimum link) path from a point $x$ to a point $y$ inside $P$
\item [{$\bullet$}]  $|X|:$ the link distance of a minimum link path $X$
\item [{$\bullet$}]  $n(X):$ the number of members of a set $X$
\item [{$\bullet$}]  $Pocket(x):$ invisible regions of $P$ from a point $x$, which are separated from $V(x)$
\item [{$\bullet$}]  $SPM(x):$ the \emph{shortest path map} (window partition) of $P$ with respect to a point or
                     line segment $x$ \cite {Suri_1990}
\end{enumerate}

More precisely, we are looking for a minimum link path between $s$ and $t$ that should have
a non-empty intersection with $V(q)$ for both cases.
Each region of $Pocket(q)$ has exactly one edge in common with $V(q)$, called an \emph{edge} of $V(q)$.
However, there is no need to compute all edges of $V(q)$.
Indeed, a single edge $e_q$ of $V(q)$ would be sufficient to find a $q$-$visible$ path (not any such edge, see Section 4.1).

Without preprocessing the query can be answered in linear time \cite {Zarrabi_2019}, but here
our goal is to achieve a logarithmic query time using a preprocessing
(this would be optimal as indicated in \cite {Arkin_1995}).
We define two types of output for either case mentioned above: one to find out $\lvert MLP(s,t,q) \rvert$ in
$O(\log n)$ time, and another to report $MLP(s,t,q)$ in $O(m+\log n)$ time, where $m=\lvert MLP(s,t,q) \rvert$.

\section{Basic lemmas and definitions}
A \emph{planar subdivision} is a partition of the 2-dimensional plane
into finitely many vertices, edges and faces $(V, E, F)$.
One type of planar subdivision is \emph{simply connected planar subdivision} \emph{(SCPS)} with an additional restriction:
any closed paths lying completely in one region of a SCPS can
be topologically contracted to a point. Therefore, a region in a SCPS cannot contain any other regions \cite {Finke_1995}.
According to \emph{Euler}'s formula for SCPS, we have the following equation: $n(V)-n(E)+n(F)=2$.

As stated above, the notion of SPM introduced in \cite {Suri_1990} is central to our discussion.
Indeed, $SPM(x)$ denotes the SCPS of $P$ into \emph{faces}
with the same link distance to a point or line segment $x$.
$SPM(x)$ has an associated set of \emph{windows}, which are chords of $P$ that serve as boundaries between
adjacent faces.
Starting and end points of a window $w$ will be denoted by $\alpha(w)$ and $\beta(w)$, respectively.
We have the following two lemmas as the basic fundamental properties of SPM:

\begin{lemma}\label{Linear}
For any point or line segment $x$ $\in$ $P$, $SPM(x)$ can be constructed in $O(n)$ time.
\end{lemma}

\begin{proof}
This follows from Theorem 1 in \cite {Suri_1990} and Chazelle's linear time triangulation algorithm \cite {Chazelle_1991}.
\end{proof}

\begin{lemma}\label{Intersection}
Given a point or line segment $x$ $\in$ $P$,
any line segment $L$ intersects at most three faces of $SPM(x)$
(see the proof of Lemma 3 in \cite {Arkin_1995}).
\end{lemma}

The value of each face of $SPM(x)$ is defined as the link distance from $x$ to any points of that face.
These values are added to each face during the construction of $SPM(x)$ in Lemma~\ref{Linear}.
Let $||SPM(x)||$ denote the number of faces of $SPM(x)$.
Also, let $F_x(i)$ be a face of $SPM(x)$ and let $||F_x(i)||$ denote the value of $F_x(i)$,
where $1 \leq i \leq ||SPM(x)||$
(see Figure~\ref{fig:SPM}).

\begin{figure}[h]
	\centering
	\includegraphics[width=13.8cm,keepaspectratio=true]{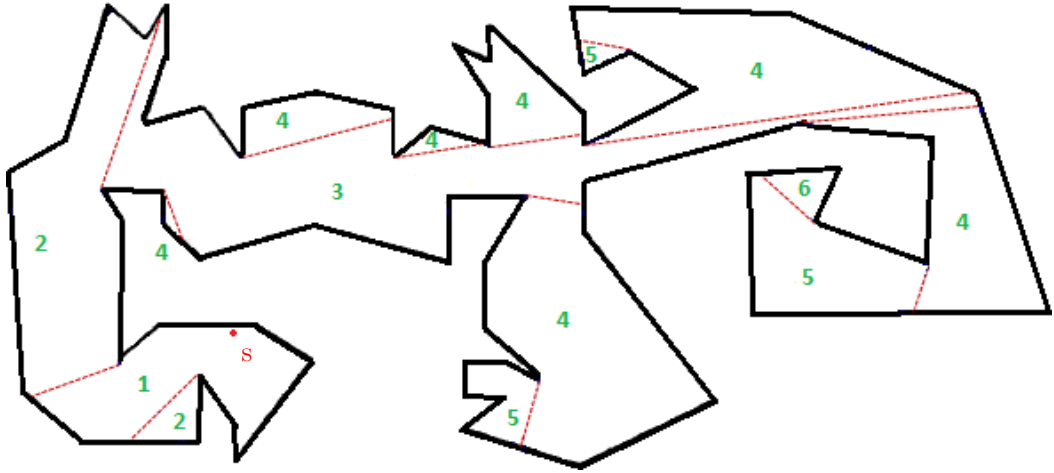}
	\caption{The value of each face $F_s(i)$ of $SPM(s)$, where $1 \leq i \leq ||SPM(s)||=15$}
	\label{fig:SPM}\vspace*{-2mm}
\end{figure}

\medskip
The \emph{window tree} $WT(x)$ denotes the planar dual of $SPM(x)$.
It has a node for each face and an arc between two nodes if their faces share an edge.
$WT(x)$ is rooted at $x$ and each node of it is labeled with a window.
The computation of $WT(x)$ takes $O(n)$ time \cite {Suri_1990}.
According to \cite {Suri_1990}, $F_x(i)$ is generated by its corresponding window $w_i \in WT(x)$ and vice versa, i.e.,
we define $G_x(w_i)=F_x(i)$ and $w_i=G^{-1}_x(F_x(i))$ ($1 \leq i \leq ||SPM(x)||$).
If $i$ is not specified, we define $G_x(w(x))=F_x$ and $w(x)=G^{-1}_x(F_x)$, where
$w(x)$ is a window of $WT(x)$ and $F_x$ is the corresponding face of $SPM(x)$.

\medskip
Consider the \emph{shortest path maps} of $P$ with respect to the points $s$ and $t$.
To find the intersection of the two maps $SPM(s)$ and $SPM(t)$, the \emph{map overlay} technique is employed.
One of the well-known algorithms for this purpose was introduced by Finke and Hinrichs \cite {Finke_1995}.
The algorithm computes the overlay of two SCPSs in optimal time
$O(n + r)$, where $n$ is the total number of edges of the two SCPSs and $r$ is the number of intersections
between the edges in the worst case.
Both the input and output SCPSs are represented
by the \emph{quad view} data structure (a trapezoidal decomposition of a SCPS) \cite {Finke_1995}.

The intersection of $SPM(s)$ and $SPM(t)$ creates a new SCPS inside $P$.
We call each face of this SCPS a \emph{Cell}.
Also, $Ce$ is defined as a set such that each member of it points to a \emph{Cell}.
By the construction of \emph{Cells}, each \emph{Cell} ($Ce(i)$) is the intersection of two faces, one face from $SPM(s)$
and another face from $SPM(t)$, i.e., $Ce(i)$= $F_s(j) \cap F_t(k)$.
The value of each \emph{Cell} (the value of $Ce(i)$) is defined as $||Ce(i)||= ||F_s(j)|| + ||F_t(k)||$
($1 \leq i \leq n(Ce)$, $1 \leq j \leq ||SPM(s)||$ and $1 \leq k \leq ||SPM(t)||$).
These values are added to each \emph{Cell} during the overlaying, and will later be used for computing $|MLP(s,t,q)|$
(see Figure~\ref{fig:Cells}).

\begin{figure}[h]
	\centering
	\includegraphics[width=13.8cm,keepaspectratio=true]{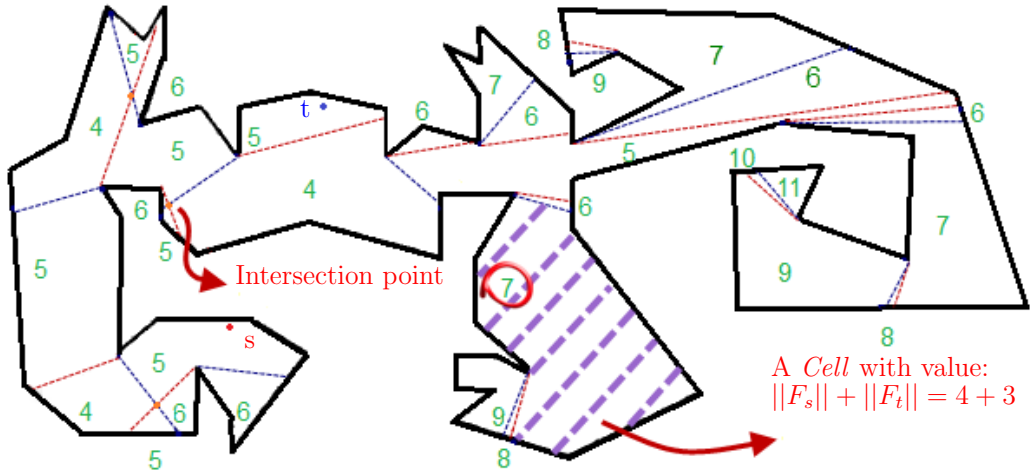}\vspace*{-1mm}
	\caption{Overlaying $SPM(s)$ and $SPM(t)$ to create \emph{Cells}, and their values}
	\label{fig:Cells}\vspace*{-2mm}
\end{figure}

Let $W_s$ and $W_t$ be the set of windows of $SPM(s)$ and $SPM(t)$, respectively, and $W=(W_s \cup W_t)$.
The following lemma shows the number of \emph{Cells} ($n(Ce)$):

\begin{lemma}\label{Cells_Num}
Overlaying $SPM(s)$ and $SPM(t)$ inside $P$ creates $n(W)+r+1$ Cells, where $r$ is the number of intersections
inside $P$ (excluding the boundary of $P$) between the windows of $W_s$ and $W_t$.
\end{lemma}

\begin{proof}
Let $(V, E, F)$ be the triple sets of the new SCPS inside $P$ after overlaying.
If we omit the outer face of $P$ from $n(V)-n(E)+n(F)=2$, we conclude: $n(Ce)=n(E)-n(V)+1$.
All the vertices and edges of the new SCPS, $(V, E)$ can be divided into two groups, $(V_1, E_1)$
on the boundary of $P$ and $(V_2, E_2)$ inside $P$, where $n(V)=n(V_1)+n(V_2)$ and $n(E)=n(E_1)+n(E_2)$.
Let $w \in W$.
For the first group, $n(V_1)$ and $n(E_1)$ are the $n$ vertices and $n$ edges of $P$, respectively,
plus the number of those $\beta(w)$ not already counted among the vertices of $P$. Thus, $n(E_1)-n(V_1)=0$.
Note that $\alpha(w)$ is always on a reflex vertex of $P$ and does not change $n(V_1)$ or $n(E_1)$.
For the second group, consider the windows of $W$.
Based on the definition of windows, there is no intersection between the windows of SPM.
Therefore, each of the $r$ intersections corresponds to only two windows, one from $W_s$ and another from $W_t$.
Suppose that there is no common window between $W_s$ and $W_t$, i.e., $n(W)=n(W_s)+n(W_t)$.
In this case, each intersection creates four segments for two distinct windows.
It is easy to deduce that by induction, the total number of edges on the windows of $W$ is $n(W_s)+n(W_t)+2r$.
But, $n(V_2)=r$ and $n(E_2)-n(V_2)=n(W_s)+n(W_t)+r=n(W)+r$.
For the case $n(W_s \cap W_t)>0$, we have to consider only one of the two coincident windows due to the fact that
they do not create a new segment. Therefore, $n(E_2)-n(V_2)=n(W_s)+n(W_t)-n(W_s \cap W_t)+r=n(W_s \cup W_t)+r$.
Finally, $n(Ce)=n(E_1)-n(V_1)+n(E_2)-n(V_2)+1=n(W)+r+1$.
\end{proof}

As depicted in the example in Figure~\ref{fig:Cells},
$n(Ce)=32$, $n(W_s)=n(W_t)=14$, $n(W_s \cup W_t)=28$ and $r=3$.
The number of the windows of SPM inside $P$ depends on the number of the reflex vertices of $P$.
Since the total internal angles of $P$ is $(n-2)*180$ degrees, $n(W_s)$ and $n(W_t) \leq n-3$.
According to Lemma~\ref{Intersection}, each window of $W_s$ intersects at most two windows
of $W_t$ and vice versa. Thus, each window in $W_s$ or $W_t$ contains at most two intersection points inside $P$, i.e.,
$r \leq 2*\min(n(W_s),n(W_t)) \leq 2(n-3)$.
Based on Lemma~\ref{Cells_Num} and the above argument, we have:\smallskip

$n(Ce)=n(W)+r+1 \leq n(W_s)+n(W_t)+r+1 \leq (n-3)+(n-3)+2(n-3)+1=4n-11=O(n)$.\smallskip

\begin{lemma}\label{Cells_Construction}
Construction of Cells and computation of their values for two points $s$ and $t$ inside $P$ can be done in $O(n)$ time.
\end{lemma}

\begin{proof}
The construction of the two maps $SPM(s)$ and $SPM(t)$ can be done in $O(n)$ time (Lemma~\ref{Linear}).
The total complexity of the two created maps inside $P$ (total number of edges and windows) is $O(n)$.
Thus, by the \emph{map overlay} technique one can construct \emph{Cells} in $O(n+r)$ time \cite {Finke_1995}, but since
$r \leq 2(n-3)$, this is $O(n)$ time.
On the other hand, $n(Ce)=O(n)$. Thus, assigning $||F_s(j)|| + ||F_t(k)||$ to $||Ce(i)||$ takes $O(n)$ time
($1 \leq i \leq n(Ce)$, $1 \leq j \leq ||SPM(s)||$ and $1 \leq k \leq ||SPM(t)||$).
\end{proof}

The \emph{visibility graph} of $P$, denoted by $V_G(P)$ is an undirected graph of the visibility relation on the
vertices of $P$. $V_G(P)$ has a node for every vertex of $P$ and an
edge for every pair of visible vertices inside $P$.
Consider a visibility graph edge $e$.
The \emph{extension points} of $e$ refer to the intersections of the boundary of $P$ with the line containing $e$.
$V_G(P)$ and its extension points can be computed in time proportional to the size of $V_G(P)$, i.e., $O(E)$,
where $E$ is the number of edges in $V_G(P)$ \cite {Hershberger_1989}.
Indeed, the main algorithm in \cite {Hershberger_1989} outputs the edges of $V_G(P)$ in sorted order
around each vertex.

\begin{lemma}\label{Rayshooting}
We can preprocess $P$ in $O(E)$ time so that the following ray shooting query can be answered in $O(\log n)$ time:
given a vertex $v$ of $P$ and a direction $\theta$, find the intersection points of the ray from $v$ in
direction $\theta$ with the boundary of $P$, where the ray exits the interior of $P$.
\end{lemma}

\begin{proof}
After computing $V_G(P)$ and its extension points in $O(E)$ time, we store the sorted edges of $V_G(P)$ according to the angle between them in some fixed directions (e.g., clockwise) about each vertex.
For a given vertex $v$, the binary search can be used to find two adjacent edges of $V_G(P)$
so that the cone defined by the two edges contains only
the ray emanating from $v$ in direction $\theta$. By finding these edges, one can specify the two
extension points on an edge $b$ of $P$. So, $b$ is used to find the intersection point of the ray in constant time.
If the ray coincides with the edges of $V_G(P)$, all intersection points are returned.
\end{proof}

A more complicated algorithm is presented in \cite {Guibas_1987} for the \emph{ray shooting} problem in a simple polygon.
This algorithm has $O(n)$ (instead of $O(E)$) preprocessing time and the same query time ($O(\log n)$).
Note that $E=\Omega(n)=O(n^2)$.
We must use this algorithm for single query (Section 5),
but Lemma~\ref{Rayshooting} suffices for triple query (Section 6).

\section{Main idea}
The following cases ($Q(x)$) may occur for a query point $q$:
\begin{enumerate}
\item [{$Q(a)$}] At least one of $s$ or $t$ is visible from $q$, i.e., $s$ or $t$ in $V(q)$.
                 In this case, $MLP(s,t,q)=\pi_L(s,t)$.

\item [{$Q(b)$}] The points $s$ and $t$ are in two different regions of $Pocket(q)$.
                 Again in this case, $MLP(s,t,q)=\pi_L(s,t)$. Since $\pi_L(s,t)$ crosses $V(q)$,
                 it is a $q$-$visible$ path.

\item [{$Q(c)$}] Both $s$ and $t$ are in the same region in $Pocket(q)$.
                 In this case, $e_q$ is the common edge between this region of $Pocket(q)$ and $V(q)$.
                 Therefore, $MLP(s,t,q)$ should have a non-empty
                 intersection with the side of $e_q$, where $q$ lies.
\end{enumerate}

\subsection{Computation of {$e_q$}}

Now, we are ready to compute $e_q$ with respect to the points $s$ and $t$ (single query and triple query).
Let $\alpha(e_q)$ and $\beta(e_q)$ be the starting and end points of $e_q$, respectively.
For any three points $s$, $t$ and $q$ inside $P$,
the last bending vertices of $\pi_E(s,q)$ and $\pi_E(t,q)$ called $v_1$ and $v_2$ (reflex vertices of $P$), respectively,
can be found in $O(\log n)$ time, after $O(n)$ time preprocessing of $P$ \cite {Arkin_1995,Guibas_1989}.
If either $v_1$ or $v_2$ does not exist, $Q(a)$ occurs.
Otherwise, if $v_1=v_2$, $Q(c)$ occurs, otherwise, $Q(b)$ occurs.
For the two cases $Q(a)$ and $Q(b)$, there is no need to compute $e_q$. But for the case $Q(c)$,
the first intersection point of the ray emanating from $\alpha(e_q)=v_1=v_2$ in direction
$\vv{q\alpha(e_q)}$ with the boundary of $P$
is specified (Lemma~\ref{Rayshooting} or \cite {Guibas_1987}).
Indeed, this intersection point is $\beta(e_q)$.
Thus, starting and end points of $e_q$ are specified and the following corollary is concluded:

\begin{corollary}\label{Eq_Computation}
Computation of $e_q$ ($\alpha(e_q)$ and $\beta(e_q)$) with respect to any two points $s$ and $t$ inside $P$ can be done in $O(\log n)$ time, after $O(E)$ or $O(n)$ time preprocessing of $P$, if necessary (case $Q(c)$).
\end{corollary}

\subsection{Case {$Q(c)$}}
The line segment $e_q$ divides $P$ into two subpolygons, only one of which contains $q$.
We define $p$ as the subpolygon containing $q$ ($e_q \in p$).
Let $Ce_p= \{Ce(i) \cap p \mid Ce(i) \cap p \neq \emptyset, 1 \leq i \leq n(Ce)\}$.
Also, $||Ce_p(j)||$ is defined to be $||Ce(i)||$, where $Ce_p(j)$ and $Ce(i)$ are the corresponding members in $Ce_p$ and $Ce$, respectively ($1 \leq j \leq n(Ce_p)$ and $1 \leq i \leq n(Ce)$).
Let $cellmin_p=\min^{n(Ce_p)}_{i=1}(||Ce_p(i)||)$ and
$Cellmin_p= \{Ce_p(i) \mid ||Ce_p(i)||=cellmin_p, 1 \leq i \leq n(Ce_p)\}$.

\begin{figure}[!h]
\vspace*{3mm}
	\centering
	\includegraphics[width=15.5cm,keepaspectratio=true]{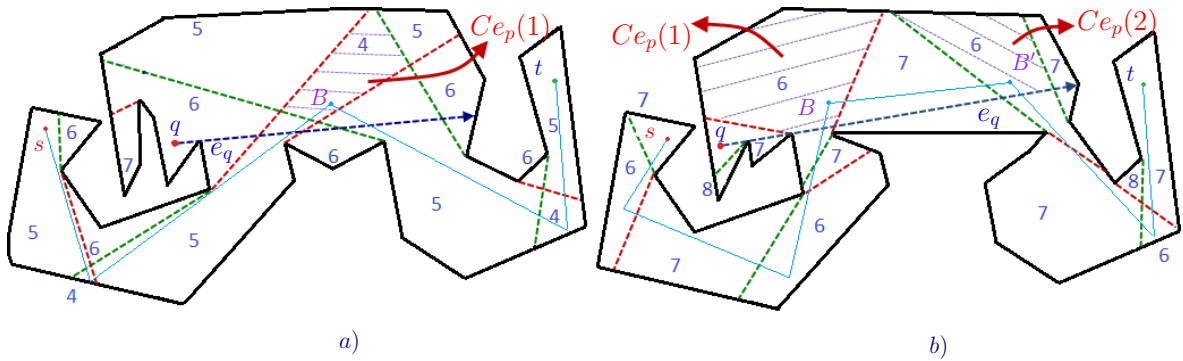}\vspace*{-2mm}
	\caption{Case $Q(c)$ and $MLP(s,t,q)$ (light blue path) for the above subcases}
	\label{fig:Bends}
\end{figure}

As mentioned above, $MLP(s,t,q) \cap p \neq \emptyset$.
According to \cite {Zarrabi_2019}, there is always a bending point $B$ (see Figure~\ref{fig:Bends}$a$)
and there are at most two bending points $B$ and $B'$
(see Figure~\ref{fig:Bends}$b$) on $MLP(s,t,q)$ inside $p$
(this follows from the triangle inequality for link distances and the fact that $s$ and $t$ lie on the same side of $e_q$,
see Lemma 3.1 in \cite {Zarrabi_2019}).
In the former case, $B$ must belong to a member of $Cellmin_p$ and
$\lvert MLP(s,t,q) \rvert=\lvert \pi_L(s,B) \rvert + \lvert \pi_L(B,t) \rvert$.
As for the latter case, $B$ and $B'$ must belong to two distinct members of $Cellmin_p$ and
$\lvert MLP(s,t,q) \rvert=\lvert \pi_L(s,B) \rvert + 1 + \lvert \pi_L(B',t) \rvert=\lvert \pi_L(s,B) \rvert + \lvert \pi_L(B,t) \rvert=\lvert \pi_L(s,B') \rvert + \lvert \pi_L(B',t) \rvert$.
Thus, for both cases, $cellmin_p=\min^{n(Ce_p)}_{i=1}\{|\pi_L(s,x)|+|\pi_L(x,t)|$ for any $x \in Ce_p(i)\}$
has been computed correctly (remember the definition of $||Ce(i)||$)
and $\min^{n(Ce_p)}_{i=1}\{|\pi_L(s,x)|+|\pi_L(x,t)|-1$ for any $x \in Ce_p(i)\}$ would not be possible for $|MLP(s,t,q)|$
(Lemma 3.1 in \cite {Zarrabi_2019}).

\subsubsection{Computation of  $Cellmin_p$  and  $cellmin_p$}

To compute $Cellmin_p$ and $cellmin_p$, we do not need to consider all members of $Ce_p$.
Indeed, we concentrate on windows of $W$ that intersect $e_q$.
As stated in Lemma~\ref{Intersection}, $e_q$ intersects
at most two windows of $W_s$ and at most two windows of $W_t$.
The status of these intersections
is verified by the position of $\alpha(e_q)$ and $\beta(e_q)$ with respect to the faces of $SPM(s)$ and $SPM(t)$.
To specify the position of $\alpha(e_q)$ and $\beta(e_q)$ with respect to the faces of $SPM(x)$
($x \in P$ is a point or line segment),
we use the \emph{point location}\footnote {For a point lying in two faces or more than one \emph{Cell},
the face or \emph{Cell} with less value is considered as the output of the \emph{point location} algorithm, respectively.}
algorithm in \cite {Edelsbrunner_1986}.
This can be accomplished by $O(n)$ time preprocessing of $SPM(x)$
(remember that $||SPM(x)||=O(n)$ and based on Lemma~\ref{Linear}, $SPM(x)$ is constructed in $O(n)$ time)
to determine which face contains $\alpha(e_q)$ and which face contains $\beta(e_q)$ in $O(\log n)$\linebreak  time.

\medskip
For $w \in W$, if ($\alpha(w)=\alpha(e_q)$ or $\alpha(w)=\beta(e_q)$) or ($\beta(w)=\beta(e_q)$),
we say that there is no intersection between $e_q$ and $w$ except for the degenerate case, where
($\beta(w)=\beta(e_q)$) and ($\alpha(e_q)$ and $\beta(e_q)$ lie in different faces, and there is no window
$w' \neq w$ from the same SPM crossing the interior of $e_q$). Note that $\beta(w)=\alpha(e_q)$ never occurs.

\medskip
Suppose that $\alpha(e_q) \in F_x(j)$ and $\beta(e_q) \in F_x(k)$.
If $j=k$, $e_q$ does not intersect any windows of $SPM(x)$.
Otherwise, if $||F_x(j)||>||F_x(k)||$ or $||F_x(j)||<||F_x(k)||$, $e_q$ only intersects the window $w_j$ or $w_k$, respectively, and if $||F_x(j)||=||F_x(k)||$, $e_q$ intersects both windows $w_j$ and $w_k$
($w_j=G^{-1}_x(F_x(j))$ and $w_k=G^{-1}_x(F_x(k))$).
After locating $w_j$ or $w_k$, if necessary,
the intersection of $e_q$ with each of them is computed in constant time.
Applying the above procedure for $SPM(s)$ and $SPM(t)$
yields at most two intersection points $s_1,s_2$ of $e_q$ with $w'(s),w''(s) \in W_s$,
and at most two intersection points $t_1,t_2$ of $e_q$ with $w'(t),w''(t) \in W_t$, respectively.
Since $n(Ce)=O(n)$ and the construction of \emph{Cells} is done in $O(n)$ time (Lemma~\ref{Cells_Construction}),
once again, we can use the \emph{point location} algorithm for \emph{Cells} \cite {Edelsbrunner_1986}.
Also, the intersection points (if any) of $w'(s),w''(s)$ with $w'(t),w''(t)$,
which are computed during the construction of \emph{Cells} can be easily determined to be inside or outside $p$.
This is done in constant time as follows:
suppose that $w'(s)$ intersects $w'(t)$ at point $I_1$. Since the position of $s_1$ and $\beta(w'(s))$ are known, the position of $I_1$ on $w'(s)$ can be determined.
If $I_1$ lies between $s_1$ and $\beta(w'(s))$, it will be inside $p$.
The same computation can be applied for the other intersection points of $w'(s),w''(s)$ with $w'(t),w''(t)$.
Let $I$ be the set of these points ($0 \leq n(I) \leq 4$).
The status of points in the set $I$ with respect to $e_q$ (inside or outside $p$)
as well as the status of the two segments $\overline{s_1s_2}$ and $\overline{t_1t_2}$ on $e_q$
(they might intersect or not)
determine $Cellmin_p$ using the \emph{point location} algorithm for \emph{Cells}
(in the degenerate cases for instance, if
$s_1$ does not exist, we replace it by $\alpha(e_q)$ or $\beta(e_q)$
depending on which one is in the same face as $s_2$).
Let $F_s$ and $F_t$ be faces of $SPM(s)$ and $SPM(t)$, respectively
(determined implicitly by the single query and triple query algorithms, but arbitrary for the moment).
For $w \in W$ and point $z$ on $w$, we define $\gamma(w,z)$
to be the point on $w$, strictly between $\beta(w)$ and $z$,
closest to $z$ among all intersections of $w$ with other windows of $W$.
If $w$ does not intersect other windows of $W$, define $\gamma(w,z)$ to be $\beta(w)$.
There are the following cases $C(x)$ (see Figure~\ref{fig:Cellmin})
for $\alpha(e_q)$ and $\beta(e_q)$ to compute $Cellmin_p$ and $cellmin_p$
(we use this fact: the intersection of two faces inside $P$ creates a \emph{simply connected region}
since $P$ does not contain any \emph{holes}):

\begin{figure}[!htbp]
	\centering
	\includegraphics[width=11.75cm,keepaspectratio=true]{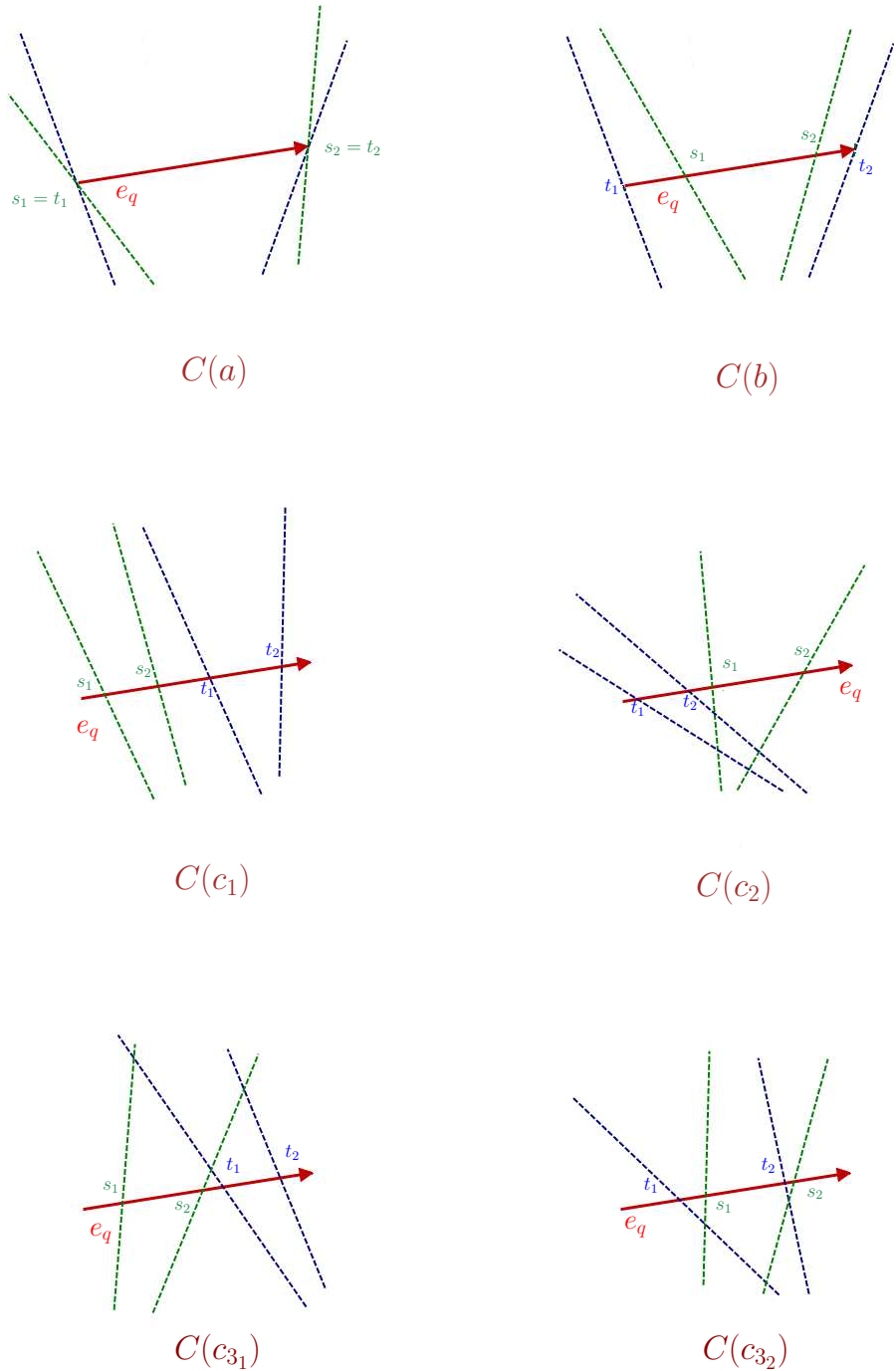}
	\caption{Cases $C(x)$ for two faces $F_s$ (green lines) and $F_t$ (blue lines), where each line may be a chain of line segments (windows or boundary of $P$). The status of $\overline{s_1s_2} \cap \overline{t_1t_2}$ as well as the status of the points in $I$ indicate each of the cases}
	\label{fig:Cellmin}
\end{figure}

\begin{enumerate}
\item [{$C(a)$}] $\alpha(e_q)$ and $\beta(e_q)$ are both in $F_s \cap F_t$.
                 In this case, $e_q$ entirely lies in one \emph{Cell} and does not intersect
                 any window(s) of $W$.
                 The portion of this \emph{Cell},
                 which lies in $p$ (like $Ce_p(1)$) is the only member of $Cellmin_p$
                 due to the fact that
                 walking from $Ce_p(1)$ to other $Ce_p(i)$ for $i>1$
                 increases $cellmin_p=||Ce_p(1)||$ by at least one.
                 To find $Ce_p(1)$, we locate $\alpha(e_q)$ and $\beta(e_q)$ in \emph{Cell} $C$.
                 Consider $w \in W$, where $\alpha(e_q)=\alpha(w)$ and
                 $\gamma(w,\alpha(w)) \in C$
                 (there is always such a window since on $\overline{q\alpha(e_q)}$ only $\alpha(e_q)$
                 is visible from $G^{-1}_s(F_s)$ and $G^{-1}_t(F_t)$).
                 Cut $C$ by $e_q$ and let $Ce_p(1)$ be the portion of $C$, where $\gamma(w,\alpha(w))$ is located.

\item [{$C(b)$}] $\alpha(e_q)$ and $\beta(e_q)$ are both either in face $F_s$ or $F_t$, but not both.
                 Therefore, $e_q$ entirely lies in one face and intersects
                 window(s) of $W_s$ or $W_t$, but not both.
                 Without loss of generality,
                 suppose that $e_q$ is in $F_s$ and intersects
                 window(s), $(w'(t)$ or $w''(t)) \in W_t$.
                 Let $w(s) = G^{-1}_s(F_s)$ and $w(t) \in W_t$ be the parent of $(w'(t)$ or $w''(t))$ in $WT(t)$.
                 Also, $F_t=G_t(w(t))$.
                 Obviously, $e_q$ is completely visible from $w(s)$ and a portion of $e_q$ is
                 visible from $w(t)$.
                 Thus, $F_s \cap F_t \cap p \neq \emptyset$. Like the argument in $C(a)$,
                 $Ce_p(1)=F_s \cap F_t \cap p$ is the only member of $Cellmin_p$ with $cellmin_p=||Ce_p(1)||$.
                 To find $Ce_p(1)$, we locate $(t_1$ and $t_2)$ in \emph{Cell} $C$ (degenerate cases included).
                 Cut $C$ by $\overline{t_1t_2}$
                 and let $Ce_p(1)$ be the portion of $C$, where $\gamma(w'(t),t_1)$ or $\gamma(w''(t),t_2)$ are located.

\item [{$C(c)$}] $\alpha(e_q)$ and $\beta(e_q)$ are both in
                 different faces of $SPM(s)$ as well as different faces of $SPM(t)$.
                 In this case, $e_q$ intersects at least one and at most two windows of $W_s$ and $W_t$.
                 Suppose that $e_q$ intersects
                 windows $(w'(s)$ or $w''(s)) \in W_s$ and $(w'(t)$ or $w''(t)) \in W_t$.
                 Indeed, $\overline{s_1s_2}$ and $\overline{t_1t_2}$ exist even for degenerate cases.
                 Let $w(s) \in W_s$ be the parent of $(w'(s)$ or $w''(s))$ in $WT(s)$ and
                 $w(t) \in W_t$ be the parent of $(w'(t)$ or $w''(t))$ in $WT(t)$.
                 Furthermore, $F_s=G_s(w(s))$ and $F_t=G_t(w(t))$.
                 Accordingly, the following cases may occur:
                 \begin{enumerate}
                 \item [{$C(c_1)$}]
                 Let $F_s \cap F_t=\emptyset$
                 (unlike $C(a)$ and $C(b)$).
                 Suppose that $w'(t)$ is closer to $F_s$ than $w''(t)$, and $F'_t=G_t(w'(t))$.
                 Since $e_q$ intersects $F_s$ and $w'(t)$, $F_s \cap F'_t \cap p \neq \emptyset$.
                 The same holds for $F_t$, $F'_s=G_s(w'(s))$ and $p$.
                 Let $Ce_p(1)=F_s \cap F'_t \cap p$ and $Ce_p(2)=F_t \cap F'_s \cap p$.
                 For any points $X_s \in F_s$, $X_t \in F_t$ and $x \in \{p-F_s-F_t\}$, we have
                 $\lvert \pi_L(s,X_s) \rvert < \lvert \pi_L(s,x) \rvert$ and
                 $\lvert \pi_L(t,X_t) \rvert < \lvert \pi_L(t,x) \rvert$.
                 Thus,
                 $\lvert \pi_L(s,X_s) \rvert + \lvert \pi_L(t,X_t) \rvert
                 < \lvert \pi_L(s,x) \rvert + \lvert \pi_L(t,x) \rvert - 1$.
                 Suppose that $x_s \in Ce_p(1) \subseteq F_s \cap p$ and $x_t \in Ce_p(2) \subseteq F_t \cap p$.
                 By the construction of $Ce_p(1)$ for any $x_s$,
                 there is a point $X_t \in F_t$ such that $x_s$ and $X_t$
                 are visible to each other. Since $F_s \cap F_t= \emptyset$,
                 $\lvert \pi_L(s,x_s) \rvert + \lvert \pi_L(t,x_s) \rvert=
                  \lvert \pi_L(s,x_s) \rvert + \lvert \pi_L(t,X_t) \rvert + 1
                  < \lvert \pi_L(s,x) \rvert + \lvert \pi_L(t,x) \rvert$.
                 The same holds, if $x_t \in Ce_p(2)$, i.e.,
                 $\lvert \pi_L(s,x_t) \rvert + \lvert \pi_L(t,x_t) \rvert
                  < \lvert \pi_L(s,x) \rvert + \lvert \pi_L(t,x) \rvert$.
                 Let $x \in \{(F_s \cap p)-Ce_p(1)\}$.
                 Since $F_s \cap p$ and $F_t$ are disjoint sets, for any point $X_t \in F_t$,
                 $\lvert \pi_L(x,X_t) \rvert \geq 2$.
                 Similarly, $\lvert \pi_L(x,X_s) \rvert \geq 2$ for $x \in \{(F_t \cap p)-Ce_p(2)\}$
                 and any point $X_s \in F_s$.
                 This indicates that $Ce_p(1)$ and $Ce_p(2)$ are the only members of $Cellmin_p$ with
                 $cellmin_p=||Ce_p(1)||=||Ce_p(2)||$.
                 In this case,
                 $\overline{s_1s_2} \cap \overline{t_1t_2}=\emptyset$ and $n(I)=0$.
                 Conversely, since $p$ and $P-p$ do not contain any holes, and $e_q$ crosses both $F_s$ and $F_t$,
                 it is easy to show that $F_s \cap F_t \cap p = \emptyset$ and
                 $F_s \cap F_t \cap (P-p) = \emptyset$.
                 Thus, $F_s \cap F_t = \emptyset$.
                 To find $Ce_p(1)$ and $Ce_p(2)$, we locate $(s_1$ and $s_2)$ in \emph{Cell} $C_1$
                 and $(t_1$ and $t_2)$ in \emph{Cell} $C_2$, respectively.
                 Cut $C_1$ by $\overline{s_1s_2}$ and cut $C_2$ by $\overline{t_1t_2}$.
                 Let $Ce_p(1)$ be the portion of $C_1$, where $\gamma(w'(s),s_1)$ or $\gamma(w''(s),s_2)$ are located.
                 Also, let $Ce_p(2)$ be the portion of $C_2$,
                 where $\gamma(w'(t),t_1)$ or $\gamma(w''(t),t_2)$ are located.
\eject
                 \item [{$C(c_2)$}]
                 Let $F_s \cap F_t \neq \emptyset$,
                 but $F_s \cap F_t \cap p=\emptyset$. In this case,
                 $F_s \cap F_t$ must be in $P-p$ between the intersections of
                 the windows $w'(s),w''(s)$ with $w'(t),w''(t)$.
                 Otherwise, since $e_q$ crosses these windows, $P-p$ contains a hole.
                 Thus, $F_s$ is divided by at least one window of $F_t$ (like $w'(t)$), where a portion of
                 $w'(t)$ is visible from $w(s)$.
                 This shows that $F_s \cap p$ would be completely visible from $w'(t)$.
                 The same holds for $F_t \cap p$ and $w'(s)$.
                 Therefore, like the case $C(c_1)$, $Ce_p(1)=F_s \cap p$ and $Ce_p(2)=F_t \cap p$
                 are the only members of $Cellmin_p$ with $cellmin_p=||Ce_p(1)||=||Ce_p(2)||$.
                 In this case,
                 $\overline{s_1s_2} \cap \overline{t_1t_2}=\emptyset$, $n(I)>0$ and $I \subset P-p$.
                 Also, the reverse situation holds.
                 The computation of $Ce_p(1)$ and $Ce_p(2)$ is similar to the case $C(c_1)$.

                 \item [{$C(c_3)$}]
                 Now, let $F_s \cap F_t \cap p \neq \emptyset$.
                 Similar to the case $C(b)$, $Ce_p(1)=F_s \cap F_t \cap p$
                 would be the only member of $Cellmin_p$ with $cellmin_p=||Ce_p(1)||$.
                 The following cases can be considered:

                 \begin{enumerate}
                 \item [{$C(c_{3_1})$}] $\overline{s_1s_2} \cap \overline{t_1t_2} = \emptyset$,
                                        $n(I)>0$ and $I \subset p$
                                        (the reverse situation holds similar to the case $C(c_2)$).
                                        If $I_1 \in I$, we locate $I_1$ in \emph{Cell} $C$ and
                                        let $Ce_p(1)=C$.

                 \item [{$C(c_{3_2})$}] Let $\overline{ST} = \overline{s_1s_2} \cap \overline{t_1t_2} \neq \emptyset$.
                                        Locate $(S$ and $T)$ in \emph{Cell} $C$.
                                        Cut $C$ by $\overline{ST}$
                                        and let $Ce_p(1)$ be the portion of $C$,
                                        where $\gamma(w,S)$ or $\gamma(w',T)$ are located.
                                        Note that
                                        $(w=w'(s)$ or $w''(s))$ and $(w'=w'(t)$ or $w''(t))$,
                                        where $w,w' \in C$
                                        (if $S=T$\footnote{If
                                        the two windows that cross at this point coincide,
                                        $Ce_p(1)$ will be the portion of this window
                                        that lies in $p$.}
                                        and ($\gamma(w,S)$ and $\gamma(w',T)$ belong to different
                                        \emph{Cells}), $Ce_p(1)=S=T$).
                 \end{enumerate}
                 \end{enumerate}
\end{enumerate}

The above cases\footnote{For the cases, where the computation of $I$ is required,
windows $w$ with ($\alpha(w)=\alpha(e_q)$ or $\beta(e_q)$) or ($\beta(w)=\beta(e_q)$) are not needed.
This is compatible with the definition of points intersecting $e_q$ (Section 4.2.1).}
for $C(x)$ indicate that $1 \leq n(Cellmin_p) \leq 2$ with the value $cellmin_p$
for all $Cellmin_p$ members.
It is easy to see that the most expensive part of these computations takes $O(\log n)$ time
(others require constant time).
The following corollary summarizes the above argument:

\begin{corollary}\label{Cellmin_Computation}
$Cellmin_p$ and $cellmin_p$ for the given $e_q$ are computed in $O(\log n)$ time,
after $O(n)$ time preprocessing of $SPM(s)$, $SPM(t)$ and Cells.
\end{corollary}

\subsection{Algorithm overview}
For any points $s$, $t$ and $q$ inside $P$, the main algorithm flow is described as follows:

\begin{enumerate}
\item [{$1)$}] Determine which cases of $Q(a)$, $Q(b)$ or $Q(c)$ occur.

\begin{enumerate}
\item [{$1.1)$}] In case $Q(a)$ or $Q(b)$ just compute $\pi_L(s,t)$ and return.

\item [{$1.2)$}] In case $Q(c)$ compute $e_q$ and continue.
\end{enumerate}

\item [{$2)$}] For single query compute $Cellmin_p$ and select an arbitrary point $x$ from a member of $Cellmin_p$.

\item [{$3)$}] For triple query select an arbitrary point $x$ from a portion of $e_q$ (determined in Section 6) or let $x \in I$ (remember the definition of $I$ from Section 4.2.1).

\item [{$4)$}] Report $\pi_L(s,x)$ appended by $\pi_L(x,t)$.
\end{enumerate}

The following sections illustrate how to compute $MLP(s,t,q)$ in detail.

\section{Single query}
We first describe the \emph{preprocessing} phase of the algorithm for two given points $s$, $t$ and a query point $q$ inside $P$:

\begin{enumerate}
\item [{$P_1)$}] Build a data structure for answering shortest Euclidean
                 path queries between two arbitrary points
                 inside $P$.

\item [{$P_2)$}] Construct $SPM(s)$, $SPM(t)$
                 and compute the value of each face of these maps. Also, construct the window trees
                 $WT(s)$ and $WT(t)$ ($G_s$, $G^{-1}_s$ and $G_t$, $G^{-1}_t$).

\item [{$P_3)$}] Construct \emph{Cells} and compute their values.

\item [{$P_4)$}] Prepare $SPM(s)$, $SPM(t)$ and \emph{Cells} for \emph{point location} queries.

\item [{$P_5)$}] Build a data structure for \emph{ray shooting} queries inside $P$.
\end{enumerate}

The \emph{query processing} algorithm for computing $\lvert MLP(s,t,q) \rvert$ proceeds as follows
(note that $Flag_s$ and $Flag_t$ indicate the number of windows of each face intersecting $e_q$):
\begin{enumerate}
\item [{$Q_1)$}] Compute implicit representations of $\pi_E(s,q)$ and $\pi_E(t,q)$, and extract from them the last
                 vertices $v_1$ and $v_2$, respectively. This can be done by the data structure of Step $P_1$.

\item [{$Q_2)$}] If either $v_1$ or $v_2$ does not exist, report $\lvert \pi_L(s,t) \rvert$ (case $Q(a)$),
                 as in \cite {Suri_1990} using Steps $P_1$, $P_2$ and $P_4$.

\item [{$Q_3)$}] If $v_1 \neq v_2$, again report $\lvert \pi_L(s,t) \rvert$ (case $Q(b)$), otherwise,
                 compute $e_q$ using the data structure of Step $P_5$ (case $Q(c)$) and continue the following steps:

\item [{$Q_4)$}] Locate $\alpha(e_q)$ and $\beta(e_q)$ in the faces of
                 $SPM(s)$ to compute two faces $F_{s}(j)$ and $F_{s}(k)$
                 containing them, respectively.
                 This can be done by Steps $P_2$ and $P_4$ ($1 \leq j,k \leq ||SPM(s)||$).

\item [{$Q_5)$}] Let $w_j=w_k=Null$.

                 If $j=k$, $\{Flag_s=0\}$,
                 otherwise,

                 $\{$
                 if $||F_{s}(j)|| > ||F_{s}(k)||$,
                 $\{Flag_s=1$; $w_j=G^{-1}_s(F_{s}(j))$; $s_1=w_j \cap e_q$;
                 $s_2=(\alpha(e_q)$ or $\beta(e_q))$ depending on which one belongs to $F_{s}(k)\}$,
                 otherwise,
                 if $||F_{s}(j)|| < ||F_{s}(k)||$, $\{Flag_s=1$; $w_k=G^{-1}_s(F_{s}(k))$;
                 $s_1=(\alpha(e_q)$ or $\beta(e_q))$ depending on which one belongs to $F_{s}(j)$; $s_2=w_k \cap e_q\}$,
                 otherwise, $\{Flag_s=2$; $w_j=G^{-1}_s(F_{s}(j))$; $w_k=G^{-1}_s(F_{s}(k))$;
                 $s_1=w_j \cap e_q$; $s_2=w_k \cap e_q\}$
                 $\}$.

                 Let $w'(s)=w_j$, $w''(s)=w_k$.

\item [{$Q_6)$}] Repeat Steps $Q_4$ and $Q_5$ for $SPM(t)$ to compute $\{Flag_t$, $w'(t)$, $w''(t)$, $t_1$, $t_2\}$.

\item [{$Q_7)$}] If $(Flag_s>0$ and $Flag_t>0)$, compute the set $I$ and $\overline{s_1s_2} \cap \overline{t_1t_2}$.

\item [{$Q_8)$}] If $Flag_s=Flag_t=0$, do as $C(a)$, otherwise,
                 if $Flag_s*Flag_t=0$, do as $C(b)$,
                 otherwise,
                 if $\overline{s_1s_2} \cap \overline{t_1t_2} \neq \emptyset$, do as $C(c_{3_2})$,
                 otherwise, if $n(I)=0$, do as $C(c_1)$, otherwise,
                 if $I \subset p$ (checkable in constant time),
                 do as $C(c_{3_1})$, otherwise, do as $C(c_2)$.
                 This can be done by Steps $P_2$, $P_3$ and $P_4$.
                 For all these conditions report $cellmin_p$.
\end{enumerate}

To analyse the time complexity of this algorithm,
observe that all of the above preprocessing requires altogether $O(n)$ time.
This follows from \cite {Guibas_1989}, \cite {Suri_1990}, Lemma~\ref{Cells_Construction}, \cite {Edelsbrunner_1986} and
\cite {Guibas_1987} for Steps $P_1$, $P_2$, $P_3$, $P_4$ and $P_5$, respectively.
On the other hand, based on Corollary~\ref{Eq_Computation} (with $O(n)$ preprocessing), Corollary~\ref{Cellmin_Computation}
and \cite {Suri_1990}, the query processing phase can be done in $O(\log n)$ time.

For computing $MLP(s,t,q)$, we must report $\pi_L(s,t)$ instead of $\lvert \pi_L(s,t) \rvert$
in Steps $Q_2$ and $Q_3$.
Also, Step $Q_8$ must be modified to report $\pi_L(s,x)$ appended by $\pi_L(x,t)$,
where $x$ is a point in a member of $Cellmin_p$
(in cases $C(c_1)$ and $C(c_2)$, $n(Cellmin_p)=2$).
Thus, the following theorem is proved:

\begin{theorem}
Given a simple polygon $P$ with $n$ vertices and two points $s$, $t$ inside it,
we can preprocess $P$ in time $O(n)$ so that for a query point $q$, one can find $\lvert MLP(s,t,q) \rvert$
in $O(\log n)$ time. Further, $MLP(s,t,q)$ can be reported in an additional time $O(\lvert MLP(s,t,q) \rvert)$.
\end{theorem}

\section{Triple query}
In this section, we propose an algorithm for three query points.
Our method is closely related to the work of Arkin et al. \cite {Arkin_1995}
(note that the construction of $V_G(P)$ is essential in \cite {Arkin_1995}).
So, we borrow the related
terminology from \cite {Aggarwal_1989, Arkin_1995} and review some terms adapted
to the notation used in this paper.

Consider a window $w$ of $SPM(x)$, where $x \in P$ is a point or line segment.
The \emph{combinatorial type} of $w$ is the vertex-edge pair $(v,e)$,
where $v=\alpha(w)$ is a reflex vertex of $P$ and $\beta(w)$
lies on an edge $e$ of $P$.
The combinatorial type of $SPM(x)$ is a listing of the combinatorial types of all of its windows.
Constructing $SPM(x)$ for all vertices of $P$ and all extension points of $V_G(P)$ edges
creates a list of windows.
The endpoints of these windows on the boundary of $P$ together with the vertices of $P$ partition the boundary of $P$ into
$O(n^2)$ intervals, called \emph{atomic segments}.
We can sort the endpoints of all windows along the boundary edges of $P$ in $O(n^2 \log n)$ time
and get access to an ordered list of atomic segments on each edge of $P$.
The following lemma from \cite {Arkin_1995} describes the special characteristic of atomic segments:

\begin{lemma}\label{Atomic}
If $L$ is an atomic segment on the boundary of $P$, the combinatorial type of $SPM(x)$ is the same for all points
$x$ in the interior of $L$.
\end{lemma}

Given a polygonal path $\Pi$ inside $P$, an interior edge $e \in \Pi$ is called a \emph{pinned} edge if it passes through
two vertices of $P$ on opposite sides of $e$ such that $e$ is tangent to $P$ at these vertices.
The \emph{greedy} minimum link path (as defined in \cite {Arkin_1995}) from a point $x$ to a point $y$ inside $P$
(called $\pi_{LG}(x,y)$ here) uses only
the extensions of the windows of $SPM(x)$
and the last link is chosen to pass through the last vertex of $\pi_E(x,y)$.
It is easy to verify whether $\pi_{LG}(x,y)$ has a pinned edge or not.
This can be done in $O(n)$ time by traversing the path $\pi_{LG}(x,y)$ for two arbitrary points $x$ and $y$ inside $P$.

On the other hand, we can check if there is a pinned edge between the atomic segment $L$ (any point $x \in L$)
and $\alpha(w_i)$ for all windows $w_i$ of $SPM(L)$ during the construction of $SPM(L)$ in $O(n)$ time.
Let $w$ be a window of $SPM(L)$ such that $\pi_{LG}(x,\alpha(w))$ has no pinned edge.
Also, let $\beta(w)$ lie on an edge $e$ of $P$.
According to Lemma~\ref{Atomic}, the combinatorial type of $w$ is the pair $(\alpha(w),e)$ for all $x \in L$.
Indeed, as $x$ varies along $L$, $\beta(w)$ varies along $e$ according to a \emph{projection function}
$f(x)$, which can be written as a fractional linear form \cite {Aggarwal_1989}:
$\beta(w)=f(x)=$
\resizebox{0.08\hsize}{!}{$\frac{Ax+B}{Cx+D}$}

The four constants $A, B, C$ and $D$ depend on the atomic segment $L$, the fixed point $\alpha(w)$
(reflex vertex of $P$) and the edge $e$.
In the case that $\pi_{LG}(x,\alpha(w))$ has a pinned edge, it is not required to compute the projection function for $L$ and $w$ as the position of $\beta(w)$ on the edge $e$ would not change when $x$ varies along $L$.
Thus, for each window $w_i$ of $SPM(L)$, where there is no pinned edge on
$\pi_{LG}(x,\alpha(w_i))$ for any $x \in L$, we compute and store the projection function $f_i$.
This preprocessing can be done in $O(n)$ time for each atomic segment $L$
during the construction of $SPM(L)$ \cite {Aggarwal_1989, Arkin_1995}.
So, for a window $w_i$, we can evaluate the exact position of $\beta(w_i)$ as $x$ varies along $L$ in constant time.
Let $SPM_A(L)$ be the data structure $SPM(L)$ plus the following:
for each window $w_i$ of $SPM(L)$, if there is a pinned edge between $L$ and $\alpha(w_i)$,
set $flag(w_i)=1$, otherwise, set $flag(w_i)=0$, then compute and store the projection function $f_i$ for $w_i$.
The following corollary is concluded from the above argument and Lemma~\ref{Linear}:

\begin{corollary}\label{Additional_Info}
For an atomic segment $L$, $SPM_A(L)$ can be constructed in $O(n)$ time.
\end{corollary}

Now, we describe the proposed algorithm for three query points $s$, $t$ and $q$ inside $P$.
Unlike the algorithm developed for single query,
we only attempt to find $\overline{s_1s_2} \cap \overline{t_1t_2}$ or $I$.
Indeed, for triple query, we may need to update all the windows of a \emph{Cell}.
In the worst case, the number of these windows is $O(n)$, and hence the queries cannot be answered in $O(\log n)$ time.
On the other hand, for any $w \in W$ as $\beta(w)$ varies along an edge of $P$,
we may need to update the value of \emph{Cells} in $O(n)$ time.
For these reasons, $Cellmin_p$ and $cellmin_p$ cannot be used for triple query and
only part of the locus (lying on $e_q$ or a point in $I$) is found in $O(\log n)$ time.
We perform the following \emph{preprocessing} step on $P$:

\begin{enumerate}
\item [{$P_0)$}]
Build a data structure for answering minimum link
path queries between two arbitrary points
inside $P$ (this includes Steps $P_1$, $P_2$ and $P_4$ of the single query algorithm, and the construction of $V_G(P)$ mentioned in Lemma~\ref{Rayshooting}).
Also, with this data structure an ordered list of atomic segments on each edge of $P$ is computed.
We can modify this step for each atomic segment $L$ as follows:
construct $SPM_A(L)$ as well as the value of each face and $WT(L)$.
Also, prepare $SPM_A(L)$ for \emph{point location} queries
(note that windows of $SPM_A(L)$ are in fixed positions, but the position of windows of $SPM_A(x)$ may change
for an arbitrary point $x \in L$).
\end{enumerate}

According to \cite {Arkin_1995, Edelsbrunner_1986} and Corollary~\ref{Additional_Info},
since we have $O(n^2)$ atomic segments, the total time complexity of this step is $O(n^3)$.

For an atomic segment $L$, let $\delta(e_q,L)$ be the set of windows $w_i$ of $SPM_A(L)$,
where $\beta(e_q)$ belongs to an edge $e$ of $P$
such that the combinatorial type of $w_i$ is $(\alpha(w_i),e)$, $flag(w_i) \neq 1$ and $e_q \cap w_i=\emptyset$.
Based on Lemma~\ref{Intersection}, each edge of $P$ intersects at most two windows of $SPM_A(L)$.
Thus, $n(\delta(e_q,L)) \leq 2$, and by Step $P_0$, it can be computed in constant time.
Indeed, the purpose of introducing $\delta(e_q,L)$ is to store all the possible windows in $SPM_A(L)$ which
may intersect $e_q$ as a point $x$ varies along $L$ except those windows that are currently intersecting $e_q$.
Since the data structure for point location queries in the following algorithm is only preprocessed for $SPM_A(L)$,
the new intersecting windows with $e_q$ can be updated from $\delta(e_q,L)$ after we know $x$ and locate the proximity of the intersecting windows.
In the case that $\beta(w_i)$ coincides with $\beta(e_q)$ after it varied along $e$,
the intersection point can be ignored (based on the definition of points intersecting $e_q$ in Section 4.2.1).
Note that $\beta(w_i)$ never coincides with $\alpha(e_q)$. Thus,
the contribution of $\alpha(e_q)$ is omitted from $\delta(e_q,L)$.

\medskip
The \emph{query processing} algorithm for computing $\lvert MLP(s,t,q) \rvert$ can be outlined as follows:
\begin{enumerate}
\item [{$Q_1)$}] Compute implicit representations of $\pi_E(s,q)$ and $\pi_E(t,q)$, and extract from them the
                 first vertices $u_1$, $u_2$ and last
                 vertices $v_1$, $v_2$, respectively.
                 This can be done by the data structure used for
                 shortest Euclidean path queries of Step $P_0$.

\item [{$Q_2)$}] If either $v_1$ or $v_2$ does not exist, report $\lvert \pi_{LG}(s,t) \rvert$ (case $Q(a)$)
                 using the data structure of Step $P_0$ for link distance queries.

\item [{$Q_3)$}] If $v_1 \neq v_2$, again report $\lvert \pi_{LG}(s,t) \rvert$ (case $Q(b)$), otherwise,
                 compute $e_q$ using $V_G(P)$ (see Corollary~\ref{Eq_Computation} with $O(E)$ preprocessing)
                 of Step $P_0$ for \emph{ray shooting} queries (case $Q(c)$) and continue the following steps:

\item [{$Q_4)$}] Compute the intersection (called $x_s$) of the extension of $\vv{u_1s}$ with an edge $e$ of $P$.
                 This can be done by $V_G(P)$ similar to Step $Q_3$.
                 Find the atomic segment $L_s$ by using binary search on $e$, where $x_s \in L_s$.
                 Similarly, $L_t$ and $x_t \in L_t$ can be computed for the extension of $\vv{u_2t}$.

\item [{$Q_5)$}] Locate $\alpha(e_q)$ and $\beta(e_q)$ in the faces of
                 $SPM_A(L_s)$ to compute two faces $F_{s}(j)$ and $F_{s}(k)$
                 containing them, respectively. This can be done by the data structure
                 used for \emph{point location} queries of Step $P_0$ ($1 \leq j,k \leq ||SPM(s)||$).

\item [{$Q_6)$}] If $j=k$, $\{w_j=w_k=Null\}$,
                 otherwise,
                 $\{$
                 if $||F_{s}(j)|| > ||F_{s}(k)||$, $\{w_j=G^{-1}_s(F_{s}(j))$; $w_k=Null\}$,
                 otherwise, if $||F_{s}(j)|| < ||F_{s}(k)||$, $\{w_j=Null$; $w_k=G^{-1}_s(F_{s}(k))\}$,
                 otherwise, $\{w_j=G^{-1}_s(F_{s}(j))$; $w_k=G^{-1}_s(F_{s}(k))\}$
                 $\}$.

\item [{$Q_7)$}] If $(w_j \neq Null$ and $flag(w_j) \neq 1)$, $\{$update $w_j$ according to $f_j$;
                 if $(w_j \cap e_q = \emptyset)$, $w_j=Null\}$.

                 If $(w_k \neq Null$ and $flag(w_k) \neq 1)$, $\{$update $w_k$ according to $f_k$;
                 if $(w_k \cap e_q = \emptyset)$, $w_k=Null\}$.

                 ($f_j$ and $f_k$ are the projection functions corresponding to
                 $w_j$ and $w_k$ as $x_s$ varies along $L_s$).

\item [{$Q_8)$}] Compute $\delta(e_q,L_s)$ and let $\Delta_s=\delta(e_q,L_s)$.

                 If $n(\Delta_s) > 0$, $\{$update the windows of $\Delta_s$, according to their
                 corresponding projection functions and position of $x_s$ on $L_s$;
                 let $\Delta_s$ be the set of these updated windows intersecting $e_q\}$.

\item [{$Q_9)$}] If $w_j \neq Null$, $\Delta_s=\Delta_s \cup w_j$.

                 If $w_k \neq Null$, $\Delta_s=\Delta_s \cup w_k$.

\item [{$Q_{10})$}] Repeat Steps $Q_5$, $Q_6$, $Q_7$, $Q_8$ and $Q_9$ for $SPM_A(L_t)$
                    to compute $\Delta_t$.

\item [{$Q_{11})$}] If $(n(\Delta_s)>0$ and $n(\Delta_t)>0)$,
                    compute the set $I$ for intersections of windows in $\Delta_s$ and $\Delta_t$.

\item [{$Q_{12})$}] If $n(\Delta_s)=1$,
                    $\{$let $w \in \Delta_s$;
                    $s_1=e_q \cap w$;
                    if $|\pi_{LG}(x_s,\alpha(e_q))|<|\pi_{LG}(x_s,\beta(e_q))|$,
                    $s_2=\alpha(e_q)$, otherwise $s_2=\beta(e_q)\}$.
                    This can be done by link distance queries of Step $P_0$.

                    If $n(\Delta_s)=2$, $\{$let $w,w' \in \Delta_s$;
                    $s_1=e_q \cap w$; $s_2=e_q \cap w'\}$.

                    Similarly, $t_1$ and $t_2$ are computed.

\item [{$Q_{13})$}] If $n(\Delta_s)=n(\Delta_t)=0$, $X=e_q$ (like $C(a)$), otherwise,
                    if $n(\Delta_s)=0$, $X=\overline{t_1t_2}$ (like $C(b)$), otherwise,
                    if $n(\Delta_t)=0$, $X=\overline{s_1s_2}$ (like $C(b)$), otherwise,
                    if $\overline{s_1s_2} \cap \overline{t_1t_2} \neq \emptyset$,
                    $X=\overline{s_1s_2} \cap \overline{t_1t_2}$ (like $C(c_{3_2})$),
                    otherwise, if $(n(I)=0$ or $I \subset P-p)$,
                    $X=(\overline{s_1s_2}$ or $\overline{t_1t_2})$ (like $C(c_1)$ or $C(c_2)$),
                    otherwise, $X$ will be a point in $I$ (like $C(c_{3_1})$).
                    For all these cases $\lvert \pi_{LG}(s,X) \rvert + \lvert \pi_{LG}(X,t) \rvert$ is reported.
                    This can be done by the data structure of Step $P_0$.
\end{enumerate}

The main difference between triple query and single query is the computation of $\Delta_s$ and $\Delta_t$.
Indeed, in Step $Q_7$, the windows $w_j$ and $w_k$ are specified for the point $x_s \in L_s$. If they intersect
$e_q$, we update them.
On the other hand, in Step $Q_8$, the windows of $SPM_A(L_s)$, which are candidates for intersection with $e_q$ are specified (including the windows that cross the endpoints of $e_q$).
These windows are added to the set $\Delta_s$ if they intersect $e_q$ after updating.
Further, in Step $Q_9$, the final $\Delta_s$ is computed ($n(\Delta_s) \leq 2$).
The other steps are similar to single query.

To analyse the time complexity of this phase of the algorithm,
it is easy to see that Steps $Q_6$, $Q_7$, $Q_8$, $Q_9$ and $Q_{11}$
can be done in constant time while others require $O(\log n)$ time (like the time complexity of single query).

For computing $MLP(s,t,q)$, we must report $\pi_{LG}(s,t)$ instead of $\lvert \pi_{LG}(s,t) \rvert$
in Steps $Q_2$ and $Q_3$.
Also, Step $Q_{13}$ must be modified to report $\pi_{LG}(s,x)$ appended by $\pi_{LG}(x,t)$,
where $x \in X$ or $x=X$.
Thus, the following theorem is proved:

\begin{theorem}
Given a simple polygon $P$ with $n$ vertices,
we can preprocess it in time $O(n^3)$ so that
for query points $s$, $t$ and $q$ inside $P$, one can find $\lvert MLP(s,t,q) \rvert$
in $O(\log n)$ time. Further, $MLP(s,t,q)$ can be reported in an additional time $O(\lvert MLP(s,t,q) \rvert)$.
\end{theorem}

\section{Conclusion}
We presented two algorithms to find a $q$-$visible$ path between two points inside a simple polygon with $n$ vertices for single query and triple query.
The proposed algorithms run with $O(n)$ and $O(n^3)$ preprocessing time
for each of the cases, respectively,
and answer a link distance query in $O(\log n)$ time for both cases.
Further, a constrained minimum link path can be reported in an additional time proportional to the number of links for either case.

One possible direction for further research on this problem is to consider the same topic in other domains such as
polygonal domains or polyhedral surfaces.
Another direction is to require the path in the query form
to visit a more complex object like a simple polygon (\emph{convex} or \emph{non-convex}), as opposed to a point.
In this case, if the shape of the query object is fixed, one can find a $Q$-$visible$ path for the object $Q$ while it translates or rotates inside a simple polygon, i.e., the desired path should have
a non-empty intersection with $Q$.

\subsection*{Acknowledgements}
The author wish to thank Dr Ali Gholami Rudi from Babol Noshirvani University and Dr Ali Rajaei from Tarbiat Modares University, Computer Sciences group,
for many pleasant discussions and valuable remarks.
Also, i would like to thank the anonymous referees for their valuable comments.

\bibliographystyle{fundam}

\end{document}